\documentclass[11pt]{amsart}
\usepackage{amsmath}
\usepackage[dvips]{epsfig}
\theoremstyle{plain}
\newtheorem{theorem}{Theorem}[section]
\newtheorem{lemma}[theorem]{Lemma}

\theoremstyle{definition}

\numberwithin{equation}{section}

\begin{document}

\title[Three Quasi-Local Masses]
{Three Quasi-Local Masses}
\author[Katz]{Neil N. Katz}
\address{Department of Mathematics\\
New York City College of Technology\\
300 Jay Street\\
Brooklyn, NY 11201} \email{nkatz@citytech.cuny.edu}
\author[Khuri]{Marcus A. Khuri}
\address{Department of Mathematics\\
Stony Brook University\\ Stony Brook, NY 11794}
\email{khuri@math.sunysb.edu}
\thanks{The second author acknowledges the support of
NSF Grant DMS-1007156 and a Sloan Research Fellowship.}
\begin{abstract}
We propose a definition of quasi-local mass based on the Penrose Inequality.
Two further definitions are given by measuring distortions of the exponential
map.
\end{abstract}
\maketitle

\section{Quasi-Local Mass from the Penrose Inequality}
\label{Introduction}

General Relativity differs from most classical field theories in
that there is no well-defined notion of energy density for the
gravitation field, as can be seen from Einstein's principle of
equivalence.  Thus at best one can only hope to calculate the
mass/energy contained within a domain, as opposed to at a point. Such
a concept is referred to as quasi-local mass, that is, a
functional which assigns a real number to each compact spacelike
hypersurface in a spacetime.  Of course there are many properties
that any such functional should satisfy under appropriate
conditions; most notable among these are the properties of
nonnegativity and rigidity (for an expanded list see \cite{ChristoYau}).
Although numerous definitions of quasi-local mass have been
proposed, most seem to possess undesirable properties, in fact
most fail the crucial test of nonnegativity.  However there is one
definition which appears to satisfy most of the required
properties, namely the mass proposed by Bartnik \cite{Bartnik}.

Bartnik's idea is to localize the ADM (or total) mass in the
following way.  Here we restrict attention to the time symmetric
case.  Let $(\Omega,h)$ be a compact 3-manifold with boundary, and
define an \textit{admissible extension} to be an
asymptotically flat 3-manifold $(M,g)$ with (or without) boundary satisfying the following
conditions: $(M,g)$ has nonnegative scalar curvature, the boundary (if nonempty) is minimal and no
other minimal surfaces exist within $(M,g)$, lastly $(\Omega,h)$ embeds isometrically into $(M,g)$.
Then Bartnik's mass is given by
\begin{equation*}
\mathcal{M}_{B}(\Omega)=\inf\{M_{ADM}(M,g)\mid (M,g) \text{ an admissible
extension of } (\Omega,h)\},
\end{equation*}
where $M_{ADM}$ is the ADM mass. A primary benefit of this construction is that nonnegativity is
achieved for free, from the positive mass theorem.  However it is
not a priori clear that this definition is nontrivial, in the
sense that the mass is nonzero whenever $\Omega$ is nonflat.  That
this is the case \cite{HuiskenIlmanen}, is a consequence of (and reason for) the no
horizons (minimal surfaces) assumption in the class of admissible extensions.
Although this mass satisfies many other desired properties, it
suffers from the apparent deficiency of being difficult to
compute.  In order to remedy this problem, Bartnik \cite{Bartnik} has proposed the following solution.
Namely, he conjectures that the infimum
$\mathcal{M}_{B}(\Omega)$ is realized by an admissible extension $(M_{0},g_{0})$
which is smooth and static vacuum outside of $\Omega$, is $C^{0,1}$ across
$\partial\Omega$, and has nonnegative scalar curvature (in the
distributional sense). Thus Bartnik's mass may be computed as the
ADM mass of the given static vacuum extension.

Recall that a Riemannian manifold $(M,g)$ is static vacuum if
there exists a potential function $u(x)>0$ which satisfies the
static vacuum Einstein equations (equivalently, the spacetime metric $-
u^{2}dt^{2}+g$ is Ricci flat):
\begin{equation}\label{e1}
\mathrm{Ric}(g)=u^{-1}\nabla^{2}u,\text{ }\text{ }\text{ }\text{
}\Delta u=0.
\end{equation}
Here $\nabla^{2}$ is the Hessian and $\Delta$ is the Laplacian with respect to
the metric $g$. The physical reasoning behind this conjecture is that once all
mass/energy has been squeezed out, there is nothing left to
support matter fields (vacuum) or gravitational dynamics (static).
Moreover significant mathematical motivation exists as well.  For
instance, if the mass minimizing extension is not static then
Corvino \cite{Corvino} has shown that there exist compactly supported metric
variations which increase the scalar curvature, so that one may
then perform an appropriate conformal deformation to lower the ADM
mass.

Ideally there should be a unique static vacuum extension, and so one must
append boundary conditions to the static vacuum equations.  The
choice of boundary conditions is dictated by the desire for the
positive mass theorem to remain valid for the complete manifold
$(\widetilde{M}_{0}\cup\Omega,g_{0}\cup h)$, where $\widetilde{M}_{0}=M_{0}-\Omega$.
This will in fact be the case (\cite{Miao1}, \cite{ShiTam}) if the metric $g_{0}$ remains
$C^{0,1}$ across the divide between the static and nonstatic
parts, and if the mean curvatures agree:
\begin{equation}\label{e2}
g_{0}|_{\partial\widetilde{M}_{0}}=h|_{\partial\Omega},\text{ }\text{
}\text{ }\text{ }H_{\partial\widetilde{M}_{0}}=H_{\partial\Omega}.
\end{equation}
Furthermore it follows from the Riccati and Gauss equations that
these Bartnik boundary conditions imply that the scalar curvature
is in fact nonnegative in the distributional sense, as demanded by
the conjecture.

The conjecture may be divided into two distinct parts.  The first
step is to establish a basic existence result for the static
vacuum equations (\ref{e1}) with boundary conditions (\ref{e2}).  Then with this
candidate mass minimizer in hand, the second step entails showing
that the given static solution actually realizes the infimum. While there
has been some progress on part one of the conjecture \cite{AndersonKhuri}, \cite{Miao2}, the
second part remains essentially uninvestigated.

It should be noted that the static metric extension conjecture (if true) shows
that the Bartnik mass depends only on the boundary geometry. Namely, let $\Sigma=\partial\Omega$,
then $\mathcal{M}_{B}(\Omega)$ depends only on the induced metric and mean curvature
on $\Sigma$. For this reason, we may write $\mathcal{M}_{B}(\Omega)=\mathcal{M}_{B}(\Sigma)$.

Although the Bartnik mass has many desirable properties it has been criticized for
potentially overestimating the true `physical' mass contained in a domain \cite{Szabados}. In this
section we propose a modification of the Bartnik mass with the aim of alleviating this
difficulty. In order to state the definition, we must first recall the Penrose inequality.
Let $(M,g)$ be an asymptotically flat Riemannian 3-manifold, with nonnegative scalar curvature,
and outermost minimal surface boundary $\partial M$. A minimal surface is called outermost if no other
minimal surface encloses it. If $|\partial M|$ denotes the area of the boundary, then the Riemannian Penrose
inequality asserts that
\begin{equation*}
M_{ADM}(M,g)\geq\sqrt{\frac{|\partial M|}{16\pi}},
\end{equation*}
and equality holds if and only if $(M,g)$ is isometric to the $t=0$ slice of the Schwarzschild spacetime.
This theorem has been proven by Huisken and Ilmanen \cite{HuiskenIlmanen} for one boundary component, and by Bray \cite{Bray} for multiple
boundary components. Just as Bartnik's mass is based on the positive mass theorem, we may use the Penrose inequality
to construct a quasi-local mass in an analogous way. Namely, with the same notion of an admissible extension, we propose
the following definition
\begin{equation*}
\mathcal{M}_{QL}(\Omega)=\inf\{M_{ADM}(M,g)-\sqrt{\frac{|\partial M|}{16\pi}}
	\mid (M,g) \text{ an admissible extension of } (\Omega,h)\}.
\end{equation*}
Notice that an admissible extension has minimal surface boundary and that no other minimal surfaces exist within the
extension, hence the boundary is in fact an outermost minimal surface and the Penrose inequality applies to show that
$\mathcal{M}_{QL}(\Omega)\geq 0$ for any $(\Omega,h)$ with nonnegative scalar curvature.
Moreover, it is immediate from the definition that monotonicity holds, in that if $\Omega_{1}\subset\Omega_{2}$ then
$\mathcal{M}_{QL}(\Omega_{1})\leq\mathcal{M}_{QL}(\Omega_{2})$. When compared with the Bartnik mass, it is apparent that
$\mathcal{M}_{QL}\leq\mathcal{M}_{B}$, which suggests that this new definition may remedy the potential tendency of
the Bartnik mass to overestimate the true value of the physical mass. It should also be observed that $\mathcal{M}_{QL}(\Omega)=0$
for any domain $\Omega$ contained within the $t=0$ slice of the Schwarzschild spacetime and lying outside of the horizon. Although
this differs in spirit from the Bartnik mass, it conforms with other masses such as the well-known Komar mass \cite{Szabados} that also
shares this property. We conjecture that $\mathcal{M}_{QL}(\Omega)=0$ if and only if $\Omega$ is a domain exterior to the horizon in
the $t=0$ slice of the Schwarzschild spacetime. This generalizes the rigidity statement for the Bartnik mass, in that Minkowski space
is replaced by the Schwarzschild solution. Furthermore, in analogy with Bartnik's conjecture, we conjecture that the infimum
$\mathcal{M}_{QL}(\Omega)$ is realized by an admissible extension $(M_{0},g_{0})$ which is smooth and static outside of $\Omega$, is $C^{0,1}$
across $\partial\Omega$, and has nonnegative scalar curvature (in the distributional sense). If true, $\mathcal{M}_{QL}(\Omega)$ may
then be computed as the ADM mass minus the mass of the boundary, of the given static extension. The extension should be unique, and
should depend only on the induced metric and mean curvature of the boundary $\Sigma=\partial\Omega$. Thus, as in the case with the
Bartnik mass, we should be able to write $\mathcal{M}_{QL}(\Omega)=\mathcal{M}_{QL}(\Sigma)$.

\section{Quasi-Local Mass from Volume Distortion}

In the remainder of this note, two further definitions of quasi-local mass will
be proposed; both definitions will be restricted to the time symmetric case only.
Although unrelated to the Bartnik mass, they share a common theme
in that each measures a distortion of the exponential map.  The first, denoted
$\underline{\mathcal{M}}$, measures volume distortion and is discussed in the present
section.  The second, denoted $\overline{\mathcal{M}}$, is discussed in the next section
and measures the distortion of a modulus of curves. In addition to sharing a common
theme, these two definition are highly related in a local manner, but differ globally.
More precisely, it turns out that inside strictly convex balls the modulus of curves is
given by the volume (Lemma 3.2). Moreover, volume is an upper bound
for the modulus. This means that the two definitions agree over small scales, but most
likely differ in the large with $\underline{\mathcal{M}}\leq\overline{\mathcal{M}}$.

We will use the following notation.  Let $(M,g)$ be a complete, time symmetric (zero extrinsic curvature),
spacelike hypersurface, of a spacetime satisfying the dominant energy condition. That is, $(M,g)$ is a
Riemannian $3$-manifold with nonnegative scalar curvature. The tangent space at $p\in M$ will be denoted
by $T_pM$, and $S_pM$ will denote its unit tangent space.  All curves
will be parameterized over the unit interval unless otherwise
specified.  For $v\in S_pM$ we denote by $\tau(v)$ the (possibly
infinite) distance to the tangential cut locus of $p$ in the direction
$v$.  We also denote
\begin{equation*}
E_pM=\{tv\mid v\in S_pM,\text{ }0\leq t<\tau(v) \}.
\end{equation*}
For $U\subset M$ we denote by $\mathrm{Vol}_g(U)$ its Riemannian volume, and for
$V\subset T_pM$ we denote by $\mathrm{Vol}(V)$ its Euclidean volume induced by
the inner product given by $g$ restricted to $T_pM$.  The metric ball
in $(M,g)$ of radius $r$ centered at $p$ will be written $B_{r}(p)$.
The ball of radius $r$ centered at $u\in T_pM$ will be denoted $B(r,u)$.

We make frequent use of the asymptotic expansion from Theorem 3.1 in
\cite{AG}, for the volume of metric balls in a Riemannian $n$-manifold:
\begin{equation}\label{e3}
\mathrm{Vol}_g(B_{r}(p)) = a_0 r^n -a_1\mathrm{Scal}_g(p)r^{n+2}+a_2\Lambda(g,p)r^{n+4}
	+O(r^{n+6}),
\end{equation}
where $a_0,a_1,a_2>0$ are constants depending only on $n$, and
$\Lambda(g,p)$ can be expressed in terms of coordinate fields
$\{X_1,\ldots,X_n\}$, chosen to be orthonormal at $p$, as follows:
\begin{eqnarray}\label{e4}
\Lambda(g,p) &=& -3\sum_{i,j,k,l}\mathrm{Riem}_g(X_i,X_j,X_k,X_l)^2
	+ 8\sum_{i,j}\mathrm{Ric}_g(X_i,X_j)^2  \\
	& & +5\mathrm{Scal}_g(p)^2-18\Delta_g\mathrm{Scal}_g(p).\nonumber
\end{eqnarray}

For any path connected, precompact, open set $\Omega\subset M$ and any
$p\in M$ we let
\begin{equation}\label{e5}
\mathcal{A}_p(\Omega) = \{U\subset E_pM\text{ open }\mid
	\overline{\mathrm{exp}_pU}\subset\Omega,\text{ } tu\in U \text{ for all }
	u\in U \text{ and } t\in[0,1]\}.
\end{equation}
Now we define the (outer) volume distortion
of the exponential map at the point $p\in\Omega$, by
\begin{equation*}
\underline{\mathcal{K}}_p(\Omega) = \sup_{U\in\mathcal{A}_p(\Omega)}\frac
	{\mathrm{Vol}(U)}{\mathrm{Vol}_g(\mathrm{exp}_pU)}.
\end{equation*}
A quasi-local mass (in the time symmetric case) should depend only on the the boundary $\Sigma=\partial\Omega$, its
induced metric $g|_{\Sigma}$, and mean curvature $H_{g}$; the triple $(\Sigma, g|_{\Sigma}, H_{g})$ will be referred to as
Bartnik boundary data. With this in mind we then propose the following definition of quasi-local mass
\begin{equation*}
\underline{\mathcal{M}}(\Sigma) = \inf_{h}
	\sup_{p\in\Omega}\log\underline{\mathcal{K}}_p(\Omega),
\end{equation*}
where in analogy with the Bartnik mass the infimum is taken over all metrics $h$
with nonnegative scalar, whose Bartnik boundary data $(\partial\Omega,h|_{\partial\Omega},H_{h})$
agrees with the given boundary data $(\Sigma,g|_{\Sigma}, H_{g})$. We conjecture that these boundary conditions guarantee that
this definition of mass is nontrivial. That is, it should hold that $\underline{\mathcal{M}}(\Sigma)>0$ unless $\Sigma$
embeds isometrically into Euclidean space in such a way that the mean curvature from the Euclidean embedding agrees with
$H_{g}$, the mean curvature from the embedding in $(M,g)$. Here we shall prove the following.

\begin{theorem}\label{thm1}
Let $(M,g)$ be a complete Riemannian 3-manifold with nonnegative scalar curvature. For any closed surface $\Sigma\subset M$ bounding a path connected
precompact domain, we have $\underline{\mathcal{M}}(\Sigma)\geq 0$. Equality
holds and the infimum is realized by a metric $h$ on a domain $\Omega$, if and
only if $(\Omega,h)$ is locally isometric
to Euclidean space.  If in addition there is a set $U\in\mathcal{A}_p(\Omega)$ for some $p\in\Omega$ such that
$\mathrm{exp}_p(U)=\Omega$, then $\Omega$ is isometric to a subset of Euclidean space.
\end{theorem}

\begin{proof}

We proceed by contradiction. Suppose that
$\underline{\mathcal{M}}(\Sigma) < 0$. This implies that there exists
a metric $h$ on $\Omega$ having the given Bartnik boundary data and with
$\mathrm{Scal}_h\geq 0$, such that
$\log\underline{\mathcal{K}}_p(\Omega) < 0$ for each $p\in\Omega$. Pick
a point $p\in\Omega$ and consider balls of radius $r>0$ about the origin in
$T_pM$. Note that for $r$ sufficiently small these balls are in
$\mathcal{A}_p(\Omega)$, and that their images under the exponential map are
also a sequence of metric balls. We then have
\begin{equation}\label{e5'}
\frac{\mathrm{Vol}(B(r,0))}{\mathrm{Vol}_h(\mathrm{exp}_p B(r,0))} =
\frac{\mathrm{Vol}(B(r,0))}{\mathrm{Vol}_h(B_{r}(p))}< 1.
\end{equation}
If $\mathrm{Scal}_h(p)>0$ then the expansion (\ref{e3}) contradicts (\ref{e5'}). Thus, either $\underline{\mathcal{M}}(\Sigma)\geq 0$
or $\mathrm{Scal}_{h}(p)=0$ for all $p\in\Omega$. We proceed with the zero scalar curvature assumption. Comparing the expansion (\ref{e3})
and the inequality (\ref{e5'}) again, yields $\Lambda(h,p)\geq 0$. Since $\Omega$ has dimension 3 the curvature tensor
can be expressed entirely in terms of the Ricci tensor:
\begin{eqnarray}
\mathrm{Riem}_h(X_i,X_j,X_i,X_j) &=& \frac{1}{2}\left[
	\mathrm{Ric}_h(X_i,X_i)+\mathrm{Ric}_h(X_j,X_j)-\mathrm{Ric}_h(X_k,X_k)
	\right] \label{e6} \\
\mathrm{Riem}_h(X_i,X_j,X_i,X_k) &=& \mathrm{Ric}_h(X_j,X_k), \label{e7}
\end{eqnarray}
where $\{X_1,X_2,X_3\}$ are coordinate fields orthonormal at $p$,
and $i,j,k\in\{1,2,3\}$ are distinct.  Since $\mathrm{Scal}_h\equiv 0$ on $\Omega$,
(\ref{e4}) gives the following expression for $\Lambda$:
\begin{eqnarray*}
\Lambda(h,p) &=& 8\left[\sum_i\mathrm{Ric}_h(X_i,X_i)^2+
	2\sum_{i<j}\mathrm{Ric}_h(X_i,X_j)^2 \right] \\
& & -3 \left[\sum_{[ijk]}\left(\mathrm{Ric}_h(X_i,X_i)+\mathrm{Ric}_h(X_j,X_j)
	-\mathrm{Ric}_h(X_k,X_k)\right)^2
	+8\sum_{j<k}\mathrm{Ric}_h(X_j,X_k)^2 \right].
\end{eqnarray*}
Here $[ijk]$ denotes summation over the three cyclic permutations of $(1,2,3)$.
Since $\Lambda(h,p)\geq 0$ for all $p\in\Omega$, it follows that
\begin{equation}\label{e8}
-\sum_i\mathrm{Ric}_h(X_i,X_i)^2
	+6\sum_{i<j}\mathrm{Ric}_h(X_i,X_i)\mathrm{Ric}_h(X_j,X_j)
	-8\sum_{i<j}\mathrm{Ric}_h(X_i,X_j)^2\geq 0.
\end{equation}
Furthermore, $\mathrm{Scal}_h(p)=0$ for all $p\in\Omega$ gives that
\begin{equation*}
0 = \left(\sum_i\mathrm{Ric}_h(X_i,X_i) \right)^2 = \sum_i\mathrm{Ric}_h(X_i,X_i)^2
	+ 2\sum_{i<j}\mathrm{Ric}_h(X_i,X_i)\mathrm{Ric}_h(X_j,X_j).
\end{equation*}
Combining this and (\ref{e8}) produces
\begin{equation*}
-4\sum_i\mathrm{Ric}_h(X_i,X_i)^2 -8\sum_{i<j}\mathrm{Ric}_h(X_i,X_j)^2\geq 0.
\end{equation*}
Since $p$ was arbitrary, the Ricci curvature is identically zero on $\Omega$.
Applying (\ref{e6}) and (\ref{e7}) gives that the full Riemann tensor
is also identically zero. Thus, $(\Omega,h)$ is locally isometric to Euclidean space and
$\log\underline{\mathcal{K}}_p(\Omega)= 0$ for all $p\in\Omega$. The conclusion of the
theorem now follows, except perhaps for the last sentence. To complete the proof note
that since $\mathrm{exp}_p|_U$ is injective for any $U\in\mathcal{A}_p(\Omega)$, if
$\mathrm{exp}_pU=\Omega$ then it is surjective as well.

\end{proof}

\section{Quasi-Local Mass from Distortion of a Modulus of Curves}

First we introduce some notation. As above $(M,g)$ will denote a complete Riemannian 3-manifold of
nonnegative scalar curvature. For any set of curves
$A\subset C^{0}([0,1],M)$, consider the set of Borel functions
\begin{equation*}
\mathcal{F}(A) = \{f:M\rightarrow [0,\infty] \mid
	\int_\gamma f\geq\mathrm{dist}_g(\gamma(0),\gamma(1))\text{ for all }
	\gamma\in A\},
\end{equation*}
and define the modulus of curves in $A$,
\begin{equation*}
\mathrm{mod}(A) = \inf\{\parallel f\parallel_g \mid f\in\mathcal{F}(A)\},
\end{equation*}
where $\parallel f\parallel_g$ is the $L^3$ norm induced by the Riemannian measure
(i.e. it is the volume of the possibly singular conformal metric $f^2g$).

Similarly, for a set of curves $B$ in $T_pM$ take the set of Borel functions
\begin{equation*}
\mathcal{F}_p(B) = \{f: T_pM\rightarrow [0,\infty]\mid
	\int_\gamma f\geq\parallel\gamma(0)-\gamma(1)\parallel_g
	\text{ for all }\gamma\in B\},
\end{equation*}
and define the modulus of curves in $B$,
\begin{equation*}
\mathrm{mod}_p(B) = \inf\{\parallel f\parallel_{(p)} \mid f\in\mathcal{F}_p(B)\} ,
\end{equation*}
where $\parallel f\parallel_{(p)}$ is the $L^3$ norm of $f$ on $T_pM$ induced by the inner
product $g$ restricted to $T_pM$.  Note that this modulus is the same
as that on curves in $\mathbb{R}^3$ under the Euclidean metric and the
subscript $p$ is there only to denote in which tangent space the curves reside.

For any set $U\subset(M,g)$ with nonempty interior, $\Gamma_U$ will be the
set of all continuous curves $\gamma$ parameterized on the unit interval with
$\gamma(t)$ in the interior of $U$ for all $0<t<1$ and
$\gamma(0),\gamma(1)\in\partial U$. Estimates of the modulus of sets of curves $\Gamma_U$
were given in \cite{nnk} Lemma 2.4, Lemma 2.5, and Theorem 2.6.  For the convenience of
the reader we summarize relevant results here in the following two lemmas.

\begin{lemma}\label{lemma1}
If $U=B_{r}(p)$ is a metric ball in $(M,g)$ and $r$ is less than the
convexity radius at $p$, then $\mathrm{mod}(\Gamma_U)=\mathrm{Vol}_g(U)$.
\end{lemma}


\begin{lemma}\label{lemma2}
Let $V$ be an open set with compact closure.
\begin{enumerate}
\item[(\textit{i})] If $V\subset B_{r}(p)$ with $r$ less than the
convexity radius at $p$, then $\mathrm{mod}(\Gamma_V) = \mathrm{Vol}_g(V)$.
\item[(\textit{ii})] If $V\subset T_pM$, then $\mathrm{mod}_p(\Gamma_V) = \mathrm{Vol}(V)$.
\end{enumerate}
\end{lemma}

Now we define a quasi-local mass from distortion of the exponential map
on the modulus on curves given above. Let $\Sigma$ be a closed surface embedded
inside $M$ and bounding a region, with Bartnik data $(\Sigma,g|_{\Sigma},H_{g})$.
For any compact, path connected Riemannian 3-manifold $(\Omega,h)$ with boundary $\Sigma$, and any $p\in \Omega$ we let
$\mathcal{A}_p(\Omega)$ be defined as in (\ref{e5}).  For each of these
collections let the (outer) modulus distortion $\overline{\mathcal{K}}_p$
of the exponential map at the point $p\in \Omega$ be given by
\begin{equation*}
\overline{\mathcal{K}}_p(\Omega) = \sup_{U\in\mathcal{A}_p(\Omega)}\frac
	{\mathrm{mod}_p(\Gamma_U)}{\mathrm{mod}(\Gamma_{\overline{\mathrm{exp}_pU}})}.
\end{equation*}
We then define the quasi-local mass for these distortions in the following way
\begin{equation*}
\overline{\mathcal{M}}(\Sigma) = \inf_h
	\sup_{p\in\Omega}\log \overline{\mathcal{K}}_p(\Omega),
\end{equation*}
where the infimum is taken over all metrics $h$
with nonnegative scalar curvature, whose Bartnik boundary data $(\partial\Omega,h|_{\partial\Omega},H_{h})$
agrees with the given boundary data $(\Sigma,g|_{\Sigma}, H_{g})$. We conjecture that these boundary conditions guarantee that
this definition of mass is nontrivial. That is, it should hold that $\overline{\mathcal{M}}(\Sigma)>0$ unless $\Sigma$
embeds isometrically into Euclidean space in such a way that the mean curvature from the Euclidean embedding agrees with
$H_{g}$, the mean curvature from the embedding in $(M,g)$. Here we shall prove the following.

\begin{theorem}
Let $(M,g)$ be a complete Riemannian $3$-manifold with nonnegative scalar curvature. For any closed surface $\Sigma\subset M$ bounding a path connected
precompact domain, we have $\overline{\mathcal{M}}(\Sigma)\geq 0$. Equality holds and the infimum is realized by
a metric $h$ on a domain $\Omega$, if and only if $(\Omega,h)$ is locally isometric to Euclidean space. Furthermore
$\underline{\mathcal{M}}(\Sigma)\leq\overline{\mathcal{M}}(\Sigma)$.
\end{theorem}

\begin{proof}

The proof is similar to that of Theorem 2.1, so we only give an outline. We
proceed by contradiction. Suppose that
$\overline{\mathcal{M}}(\Sigma) < 0$. This implies that there exists
a metric $h$ on $\Omega$ having the given Bartnik boundary data and with $\mathrm{Scal}_h\geq 0$, such that
$\log\overline{\mathcal{K}}_p(\Omega) < 0$ for each $p\in\Omega$. Pick a point $p\in\Omega$ and consider
$U=B(r,0)\subset T_pM$. By Lemma 3.1,
\begin{equation}\label{e9}
1 > \overline{\mathcal{K}}_p(\Omega) \geq
	\frac{\mathrm{mod}_p(\Gamma_U)}{\mathrm{mod}(\Gamma_{\overline{\mathrm{exp}_pU}})}
	= \frac{\mathrm{Vol}(B(r,0))}{\mathrm{Vol}_g(B_{r}(p))}.
\end{equation}
If $\mathrm{Scal}_{h}(p)>0$, this contradicts expansion (\ref{e3}). Thus, either $\overline{\mathcal{M}}(\Sigma)\geq 0$
or $\mathrm{Scal}_{h}(p)=0$ for all $p\in\Omega$. We proceed with the zero scalar curvature assumption. Comparing the expansion (\ref{e3})
and the inequality (\ref{e9}) again, yields $\Lambda(h,p)\geq 0$, where $\Lambda(h,p)$ is given in (\ref{e4}). As in the proof of Theorem 2.1, this implies that the full Riemann
tensor vanishes, and hence $(\Omega,h)$ is locally isometric to Euclidean space. In this case, the exponential map at each $p\in\Omega$
restricted to any $U\in\mathcal{A}_p(\Omega)$ is an isometry, and so applying Lemma 3.2 shows that $\overline{\mathcal{M}}(\Omega)=0$.
The conclusions of the theorem now follow, except perhaps for the last sentence. To complete the proof it is sufficient to show that
if a set $U$ has nonempty interior, then
\begin{equation*}
\mod(\Gamma_U)\leq \mathrm{Vol}_h(U).
\end{equation*}
Since the integral of $1$ along any curve $\gamma$ is greater than or equal to the
Riemannian distance between the endpoints of $\gamma$, we have that
$1\in\mathcal{F}(\Gamma_U)$ and the claim follows immediately.

\end{proof}

\bibliographystyle{amsplain}

\end{document}